\documentclass{elsarticle}
\usepackage{graphicx} % Required for inserting images
\usepackage{amssymb}
\usepackage{amsmath}
\usepackage{amsthm}
\usepackage{hyperref}
\usepackage{todonotes}
\usepackage{algorithm}
\usepackage{algpseudocodex}
\usepackage{tikz}
\usepackage{pgfplots}
\pgfplotsset{compat=1.18}
\usepackage{lscape}
\usepackage{comment}

\newtheorem{definition}{Definition}
\newtheorem{theorem}{Theorem}
\newtheorem{lemma}{Lemma}

\begin{document}

\begin{frontmatter}   

\title{Finding Minimal Clusters in st-DAGs}
\author{Ulrich Vogl, Markus Siegle}
\address{Universit\"at der Bundeswehr M\"unchen, Department of Computer Science,\\ 85577 Neubiberg, Germany\\
%{\bf Version of \today}
}
\ead{ulrich.vogl@unibw.de, markus.siegle@unibw.de}

\begin{abstract}
Directed Acylic Graphs with a single entry vertex and a single exit vertex (st-DAGs) have many applications.
For instance, they are frequently used for modelling flow problems or precedence conditions among tasks, work packages, etc..
This paper presents an algorithm for finding special types of subgraphs in such st-DAGs, called `clusters'.
Knowing the clusters of a given st-DAG is very useful during DAG analysis. 
Clusters are characterized by a kind of synchronizing behaviour at their entry border and at their exit border.
In this context, we introduce the notion of `syncpoint', a type of synchronisation point within a DAG, and for a given st-DAG we construct a second DAG, called MSP-DAG, whose edges are given by the precedence relation among maximum size syncpoints (MSPs).
Our new cluster finding algorithm searches for clusters between potential pairs of enclosing MSPs.
The efficiency of the algorithm stems from the fact that it works on the MSP-DAG, which is usually much smaller than the original st-DAG.
The paper includes a thorough complexity analysis of the algorithm's runtime, which turns out to be quadratic in the number of DAG vertices and exponential in the number of MSP-DAG vertices.
There is also a section reporting on experiments with randomly generated DAGs, which shows the practical applicability of the algorithm and confirm our theoretical findings.
\end{abstract}


\end{frontmatter}

\section{Introduction}

Directed acyclic graphs (DAGs), though conceptually simple,  are a powerful and widespread modelling formalism with a wide range of applications. They are frequently used for project planning, workflow analysis, the modelling of concurrent task execution, as dependency graphs in scheduling and syntactical program analysis, for studying causal inference, and many more areas of application.

The authors have been working with DAGs in the context of project planning, where the vertices of the DAG represent work packages and the edges represent causal dependencies between them.
An edge $(u,v)$ from vertex $u$ to vertex $v$ thus means that $v$ may only start once $u$ is completely finished.
In consequence, a vertex can only become active once {\em all} of its predecessors have finished.
Obviously, in this context, redundant edges -- that is, edges $(u,v)$ where there exists an alternative path from $u$ to $v$ -- are meaningless and can be ignored.
An important goal of project planning is to determine the likely overall execution time of a project represented by a DAG.
Since work package execution times may vary considerably depending on various factors, probability distributions are used to represent the random execution times of the individual work packages (i.e., of the vertices of the DAG).
The overall execution time is therefore also a random variable whose distribution needs to be computed.

In case the given DAG is series-parallel (for the definition of this property see Definition~\ref{def:serpar} in Section~\ref{sec:defs_and_props} below), calculating the overall execution time distribution from the individual vertex distributions can be performed by a straight-forward stepwise reduction scheme:
Iteratively applying the convolution or maximum operations on pairs (or sets) of intermediate subgraph distributions, thus effectively reducing the original DAG to a single vertex, will lead to the overall resulting distribution.
However, for DAGs which are not series-parallel, or DAGs which have subgraphs which are not series-parallel, this reduction scheme does not work, since the non-series-parallel parts of the graph cannot be resolved.
Therefore, in our earlier works \citep{vogl:2017}
and \citep{vogl:2023}, we developed the strategy of isolating the non-series-parallel parts, the so-called {\em minimal complex clusters}, within a given DAG and analyze their execution time distribution with the help of a probabilistic model checking tool such as PRISM \citep{KNP:2011}.
The cluster is then replaced by a single vertex with the same runtime distribution.
However, this type of cluster analysis is state space based and therefore expensive in terms of runtime and memory.
Hence, it is important that the clusters to be analyzed by model checking are as small as possible, which motivates the search for clusters of {\em minimal} size.
Without going further into the details of the procedure just described, this makes it clear that the precise definition of minimal series-parallelly {\em non-}reducible subgraphs (the minimal complex clusters), and the question of how to identify them efficiently within the context of a potentially large DAG, are important research questions.
This is precisely what the present paper is about.

The rest of the paper, which is based on the PhD thesis \citep{vogl:2025} but focuses more strictly on the graph-theoretic questions, is structured as follows:
After discussing some related work in Section~\ref{subsec:relwork}, we start with the necessary basics in Section~\ref{sec:defs_and_props}, which includes the definitions of the notions of {\em cluster} and {\em complex cluster}.
Section~\ref{sec:syncpoints} introduces the important concept of {\em syncpoints}, which represent -- as the name implies -- a kind of synchronisation points within DAGs, which may mark the entry or exit points of clusters.
Once all maximum-size syncpoints (all MSPs) in a given DAG are identified, a second, typically much smaller DAG, the so-called maximum syncpoint DAG (MSP-DAG), can be constructed, which plays a  key role in our approach.
In Section~\ref{sec:clustfind}, our algorithm for finding minimal clusters is presented in pseudocode and explained in some detail.
The section also contains a complexity analysis of the algorithm which shows that the overall runtime is quadratic in the number of vertices and exponential in the number of maximum syncpoints.
Section~\ref{sec:appl} presents a set of experiments, showing that the new algorithm is practicable for the analysis of st-DAGs of considerable size.
The paper concludes with Section~\ref{sec:conclusion}, which summarizes the major results and gives some pointers to possible future work.

\subsection{Related Work}
\label{subsec:relwork}

There are numerous publications dealing with the overall time-span analysis of DAGs whose vertices are equipped with stochastic execution time distributions.
Such techniques were developed, for example, for the analysis of parallel computer programs, running on multiprocessor systems, where the vertices of the DAG represent tasks (procedures, subroutines) to be executed.
In this context, the early tool PEPP \citep{hartleb:1993,klar:1995} provided a graphical user interface for model specification and a range of analysis techniques.
In the case of a series-parallel DAG, PEPP is able to reduce it to a single vertex with the same execution time distribution, by iteratively merging pairs (or sets) of vertices that are serially (or parallelly) linked (for a definition of `series-parallel' see Definition~\ref{def:serpar} below).
However, a DAG that is not completely series-parallel can only be treated exactly if all vertex execution times are exponentially distributed, and only expected overall times (not distributions) can be calculated. For the general non-series-parallel case, PEPP is only able to compute upper and lower bounds on the overall distribution.
Apart from identifying pairs (or sets) of serially or parallelly linked vertices, PEPP is not able to exploit any special structural properties of a given DAG.
In contrast, our approach of finding special structures, which we call clusters, and replacing their execution time distribution by a single vertex, is always applicable, efficient, and leads to precise overall time-span distributions.

More recently, techniques for determining response time distributions of complex workflows and the related tool EULERO have been presented \citep{carnevali:2023,carnevali:2022}.
The authors employ exact analysis in combination with stochastic bounding techniques, the latter being necessary for workflows which violate a `well-formed nesting' property, which essentially amounts to not being series-parallel.
However, the structure of the workflow model at hand is not determined by analysis, but has to be provided by the user as part of the model specification, which takes the form of a so-called hierarchical structure tree. In that way, EULERO avoids the problem of structurally analyzing a flat graph model.

From the graph-theoretic point of view,
characterizing special structural features of graphs, or isolating substructures having special properties, has received surprisingly little attention in the literature.
Of course, connectedness issues (such as finding the connected components of undirected graphs, or the strongly connected components of DAGs) are well-known problems for which efficient algorithms exist \citep{Tarjan:1972,McLendon:2005}.
%It is possible to test the strong connectivity of a graph, or to find its strongly connected components, in linear time
%
Related to the question of connectedness, in an {\em undirected} graph, there is the notion of clique (an $r$-clique is a complete subgraph with $r$ vertices), and the clique decision problem (deciding whether a graph has a clique of a certain size) is famous for being NP-hard \citep{Pardalos:1995, Wu:2015}.
Also, many different types of clustering problems have been studied intensively, where a cluster is a subset of vertices with high internal and low external connectivity (notice that our notion of `cluster', as defined in this paper, is different).
Clustering has numerous applications in data analysis for detecting patterns, enclosing communities, etc., within large networks \citep{schaeffer:2007}.

In contrast to the mentioned work, this paper focuses on the problem of finding special subgraphs within DAGs which are characterized by formal conditions on their entry resp.\ exit points.
We call such subgraphs minimal complex clusters (not to be confused with the much more general notion of clustering mentioned in the previous paragraph).
As already pointed out, this new type of cluster plays a key role in our approach to analyzing the overall execution time distributions of DAGs \citep{vogl:2017,vogl:2023}.

\section{Definitions and properties of st-DAGs}
\label{sec:defs_and_props}

\begin{definition}
\label{def:DAG}
A {\em directed acyclic graph} with one source and one target vertex {\em (st-DAG)} is a tuple $G=(V,E)$ where $V$ is a finite nonempty set of vertices and $E  \subseteq V \times V$ is a set of directed edges. There are no cycles, i.e.\ no nonempty sequences of directed edges leading back to the same vertex. The graph $G$ has a single source vertex $s$ (without any incoming edge) and a single target vertex $t$ (without any outgoing edge).
\end{definition}

According to Def.~\ref{def:DAG}, the smallest st-DAG consists of a single vertex and no edge, in which case $s=t$ and $E=\emptyset$.

\begin{definition}
\label{def:path}
In an st-DAG, an {\em edge-path (or just path, for short)} $ p$ {\em of length $l \geq 1$} is defined as a nonempty sequence of edges $(e_1,\dots,e_l)$ where $e_i=(u_i,v_i) \in E$ ($i=1,\dots,l$), such that $u_{i+1}=v_i$ ($i=1,\dots,l-1$). An {\em st-path of length $l$} is a path with $u_1=s$ and $v_l=t$.
For such an edge-path $p$, the corresponding {\em vertex-path} is the sequence of vertices $u_1,u_2, \dots, u_l, v_l$.
In an st-DAG $G=(V,E)$, an edge $e=(u,v) \in E$ is called {\em redundant} iff there exists an alternative path from $u$ to $v$ not containing $e$.
\end{definition}

In the sequel, if we write DAG we always mean st-DAG and we will always assume that it does not contain any redundant edges.
Should a given DAG have redundant edges, we can safely delete them, i.e.\ we can work on its (unique) transient reduction \citep{Aho:1972}.

\begin{definition}
\label{def:twins}
Given a DAG $G=(V,E)$, two or more vertices $v_1, v_2, \dots \in V$ are called {\em in-twins (out-twins)} iff they have the identical set of predecessor (successor) vertices.
\end{definition}

\begin{definition}
\label{def:serpar}
A DAG $G=(V,E)$ is called {\em series-parallel (or completely series-parallelly reducible)} if it can be reduced to a single vertex by a sequence of the following two reduction steps:
\begin{itemize}
    \item Two vertices $u$ and $v$, where $v$ is the only successor of $u$, and $u$ is the only predecessor of $v$, can be reduced serially to a single vertex.
    \item A set of $k$ vertices $v_1, \dots, v_k$ (for $k \geq 2$) which are in-twins as well as out-twins can be reduced parallelly to a single vertex.
\end{itemize}
\end{definition}
The notion of series-parallel graph in Def.~\ref{def:serpar} is equivalent to the one in \citep{valdes:1979} and \citep{klar:1995}, but \citep{valdes:1979} uses a constructive definition instead of ours based on reduction steps. 

\begin{definition}
\label{def:subgraph}
Given a DAG $G=(V,E)$, the tuple $G'=(V',E')$, with $V'\subseteq V$ and $E' \subseteq E$, is called a {\em subgraph} of $G$, iff $\forall e =(u,v) \in E': u \in V' \wedge v \in  V'$ (no dangling edges).
A vertex $v\in V'$ is called entry (exit) vertex of $G'$, if it has a predecessor (successor) in $V \setminus V'$, or if it is the source node $s$ (the target node $t$).
Given a set of vertices $V' \subseteq V$, the {\em induced subgraph} of $V'$ is the maximal subgraph of $G$ with vertex set $V'$. 
\end{definition}

\begin{definition}
\label{def:cluster}
Given a DAG $G=(V,E)$ and a vertex set $V' \subseteq V$, let $G'=(V',E')$ be the induced subgraph of $V'$. Let $A \subseteq V'$ resp.\ $B \subseteq V'$ (with $|A| \geq 2$ and $|B| \geq 2$ and $A \cap B = \emptyset$) be the set of entry resp.\ exit vertices of $G'$. $G'$ is called a {\em cluster} iff $A$ is a set of in-twins and $B$ is a set of out-twins. A cluster $G'$ is called {\em complex cluster} iff it cannot be reduced by any serial or parallel reduction step. A cluster is called {\em minimal cluster} if it does not possess a proper subgraph which is also a cluster. 
\end{definition}

According to Def.~\ref{def:cluster}, a cluster has at least four vertices.
Figure~\ref{fig:stDAG} shows an st-DAG which exemplifies the definitions of serial/parallel reducibility as well as the notion of a cluster.
As we explained in the introduction, minimal complex clusters play an important role during the analysis of a DAG, so we need efficient techniques for finding them.
Syncpoints, which are introduced in the next section,
will be of great help for identifying clusters.

\begin{figure}
    \centering
    \includegraphics[width=0.35\linewidth]{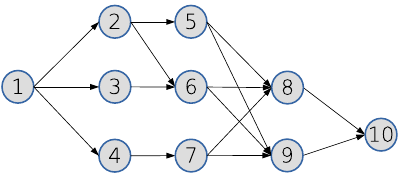}
    \caption{An st-DAG with $s=1$ and $t=10$. Vertices 4 and 7 can be reduced serially. Vertices 8 and 9 can be reduced parallelly. The subgraph induced by vertex set $\{ 2,3,5,6\}$ is a minimal complex cluster.
    The two subgraphs induced by vertex sets $\{ 2,3,4,5,6,7\}$ and $\{ 2,3,4,5,6,7,8,9\}$ are also clusters, but they are neither complex (since serial/parallel reductions inside those subgraphs are possible) nor minimal.}
    \label{fig:stDAG}
\end{figure}

\section{Syncpoints and related concepts}
\label{sec:syncpoints}

\begin{definition}
\label{def:syncpoint}
For a DAG $G=(V,E)$, consider a nonempty set of edges $\mathcal{X}=\{ e_1, \dots, e_k\} \subseteq E$, with $S = \{v \in V \; | \; v $ is end vertex of some $e \in \mathcal{X} \}$ and $P = \{v \in V \; | \; v $ is start vertex of some $e \in \mathcal{X} \}$. $\mathcal{X}$ is called a {\em Fullsyncpoint (FSP)} iff $\mathcal{X}$ contains all edges leading from $P$ to $S$ and
\begin{itemize}
\item[a)] $(|S| = 1$ or $S$ is a set of in-twins$)$ and $P$ is the (common) set of predecessors of vertices of $S$, and
\item[b)] $(|P|=1$ or $P$ is a set of out-twins$)$ and $S$ is the (common) set of successors of vertices of $P$. 
\end{itemize}
If only a) holds and $|S| \geq 2$ then $\mathcal{X}$ is called a {\em Forward-Halfsyncpoint (FHSP)}. If only b) holds and $|P| \geq 2$ then $\mathcal{X}$ is called a {\em Backward-Halfsyncpoint (BHSP)}.

For a syncpoint $\mathcal{X}$ (which is an edge set), we denote by $P_\mathcal{X}$ the set of its predecessor/start vertices and by  $S_\mathcal{X}$ the set of its successor/end vertices. 
\end{definition}

The following lemma establishes the important connection between clusters as defined in Definition~\ref{def:cluster} and syncpoints from Definition~\ref{def:syncpoint}.

\begin{lemma}
\label{lem:clust_SP}
In a DAG $G=(V,E)$, let the subgraph $G'=(V',E')$ be a cluster. Let $A \subseteq V'$ be the entry vertices of $G'$ and let $B \subseteq V'$ be the exit vertices of $G'$. Then the set of incoming edges of $A$ is a FSP or a FHSP, and the set of outgoing edges of $B$ is a FSP or a BHSP.
\end{lemma}

\begin{proof}
    According to Def.~\ref{def:cluster}, $A$ is a set of in-twins, so for the set of their incoming edges condition a) (with $S=A$ and some vertex set $P \subseteq V \setminus V'$) of the syncpoint definition (Definition~\ref{def:syncpoint}) is fulfilled. Likewise, $B$ is a set of out-twins, so for the set of their outgoing edges condition b) (with $P=B$ and some vertex set $S \subseteq V \setminus V'$) of the syncpoint definition is fulfilled.
\end{proof}

\begin{definition}
    A sycnpoint is called {\em 1-to-1 syncpoint (11SP)} iff $|S|=1$ and $|P|=1$ in Definition~\ref{def:syncpoint}. A syncpoint is called a {\em Maximum Syncpoint (MSP)} if it is not a proper subset of another syncpoint. Otherwise it is called a {\em Subsyncpoint (SSP)}.
\end{definition}

\begin{lemma}
\label{lem:maxsp}
    A Fullsyncpoint, characterized by edge set $\mathcal{X}$, is always a maximal syncpoint.
\end{lemma}

\begin{proof}
    Using the notation of Definition~\ref{def:syncpoint}, $\mathcal{X}$ contains all edges from $P$ to $S$. Suppose there was a syncpoint $\hat{\mathcal{X}}$ with $\mathcal{X} \subsetneq \hat{\mathcal{X}}$. Let $\hat{e} = (\hat{u},\hat{v}) \in \hat{\mathcal{X}} \setminus \mathcal{X}$. Then $\hat{u} \notin P$ or  $\hat{v} \notin S$ (or both). Assume wlog that  $\hat{u} \notin P$ and $\hat{v} \in S$. This contradicts Definition~\ref{def:syncpoint} a) which states that $P$ is the set of predecessors of $S$. The dual case $\hat{u} \in P$ and $\hat{v} \notin S$ leads to a similar contradiction. The case $\hat{u} \notin P$ and $\hat{v} \notin S$ contradicts the twin-conditions of both a) and b) of Definition~\ref{def:syncpoint}. 
\end{proof}

\begin{lemma}
\label{lem:sp_front_end}
 Given an st-DAG $G=(V,E)$, the set of edges leaving $s \in V$ (the unique source vertex) is a FSP. Likewise, the set of edges entering $t \in V$ (the unique target vertex) is also a FSP.
\end{lemma}

\begin{proof}
    This follows immediately from the FSP definition (Definition~\ref{def:syncpoint}).
\end{proof}

\begin{figure}
    \centering
    \includegraphics[width=0.5\linewidth]{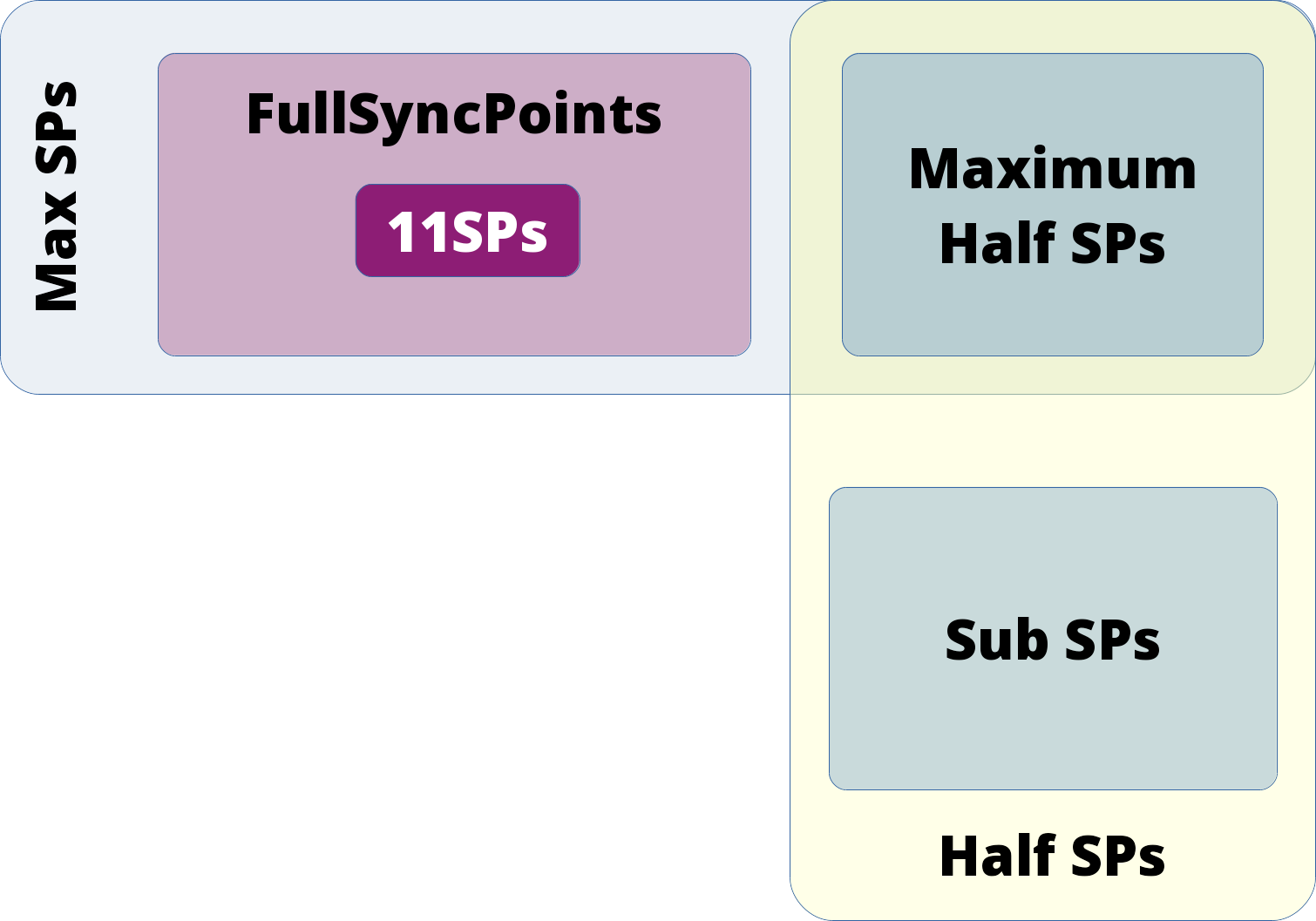}
    \caption{The relationships between different classes of syncpoints}
    \label{fig:SPsets}
\end{figure}

Figure~\ref{fig:SPsets} illustrates the relationships between the different classes of syncpoints.
Figure~\ref{fig:SPtypeExamples} shows a DAG with some examples of the different types of syncpoints.
Coloured vertical bars stretching over edge sets are used to illustrate syncpoints.
Notice that in this DAG there are two possibilities for parallel reduction (vertices 3 and 7, and vertices 4 and 6).
After performing these parallel reductions, the resulting two new vertices could be reduced serially, such that the subgraph induced by vertices $\{3,4,6,7 \}$ would be collapsed to a single vertex. However, at this point no further serial or parallel reductions would be possible.
\begin{figure}
    \centering
    \includegraphics[width=0.3\linewidth]{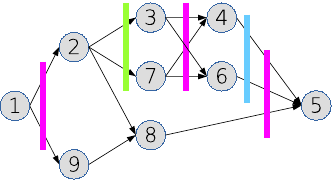}
    \caption{A DAG exemplifying different types of syncpoints. Fullsyncpoints (FSPs) are shown in magenta colour, Forward-Halfsyncpoints (FHSPs) in green and Backward-Halfsyncpoints (BHSPs) in blue. Note that the BHSP is a Subsyncpoint (SSP) of the rightmost FSP. All SPs except the BHSP are Maximum Syncpoints (MSPs).}
    \label{fig:SPtypeExamples}
\end{figure}

Algorithm~\ref{alg:findMSPs} shows an algorithm for finding all MSPs of a given st-DAG $G=(V,E)$.
It basically determines all out- resp.\ in-twin classes of $V$ and checks from each of them in forward resp.\ backward direction for possible syncpoints. Its runtime complexity can be assessed as follows:
Lines \ref{line:FindMSP1} and \ref{line:FindMSP2} can be performed in a single traversal of $G$.
If $|V| = n$ and $|E| = m < n^2$ this takes time $O(m)$.
Using suitable data structures, the for-loops in lines \ref{line:FindMSP3} and \ref{line:FindMSP9} can be performed in time $O(t)$, where $t < n$ is the number of (out- or in-)twin classes.
To see this, notice that all information needed in those two for-loops (such as the sets $S$ and $P$ in lines \ref{line:FindMSP4} and \ref{line:FindMSP10} and the sizes of all in-twin and out-twin classes) has already been precomputed in lines \ref{line:FindMSP1} and \ref{line:FindMSP2},
and that
the evaluation of the conditions in lines \ref{line:FindMSP5} and \ref{line:FindMSP11} can be realized by comparing precomputed (constant-size)
%constant-size by the method of rolling hashing
hash values of sets.
The for-loop in line \ref{line:FMSP17} can also be performed in time $O(n)$.
Therefore, the overall time complexity of Algorithm~\ref{alg:findMSPs} is $O(m)$.

\begin{algorithm}[ht]
    \caption{Algorithm $FindAllMSPs$ for finding all MSPs of an st-DAG $G=(V,E)$}\label{alg:findMSPs}
    \begin{algorithmic}[1]
        \State {\bf function} $FindAllMSPs$ ($G$: st-DAG) {\bf returns} void
        \State in a single pass of $G$, determine all out-twin classes and all in-twin classes of $V$\label{line:FindMSP1}
        \State during the same traversal of $G$, determine all vertices with out-degree one or with in-degree one \label{line:FindMSP2}
        \For{each out-twin class $P$ found} \label{line:FindMSP3}
        \State $S :=$ set of successor vertices of $P$
        \label{line:FindMSP4}
        \If{($S$ is an in-twin{} class or if $|S|==1$) $\wedge$ (in-degree of $v \in S$ is $|P|$)}
        \label{line:FindMSP5}
        \State output the set of edges leaving $P$ as a FSP
        \label{line:FSP}
        \Else
        \State output the set of edges leaving $P$ as a maximum BHSP
        \EndIf
        \EndFor
        \For{each in-twin class $S$ found} \label{line:FindMSP9}
        \State $P:=$ set of predecessor vertices of $S$
        \label{line:FindMSP10}
        \If{($P$ is an out-twin class) $\wedge$ (out-degree of $v \in P$ is $|S|$)}
        \label{line:FindMSP11}
        \State do nothing (this FSP was already found in line~\ref{line:FSP})
        \ElsIf{$|P|==1$ $\wedge$ (out-degree of $v \in P$ is $|S|$)}
        \State output the set of edges leading to $S$ as a FSP
        \Else
        \State output the set of edges leading to $S$ as a maximum FHSP
        \EndIf
        \EndFor
        \For{all vertices $u$ with out-degree one} \label{line:FMSP17}
        \If{$v :=$ (successor vertex of $u$) has in-degree one}
        \State output the edge $(u,v)$ as a 11SP
        \EndIf
        \EndFor    
    \end{algorithmic}
\end{algorithm}

\subsection{Precedence of Syncpoints}
\label{subsec:sp_prec}

For a given st-DAG $G$, we construct another DAG whose nodes are the maximum syncpoints of $G$ and whose edges are given by the natural precedence relation among syncpoints defined in Definition~\ref{def:sp_prec}.

\begin{definition}
\label{def:sp_prec}
Let $\mathcal{X} \neq \mathcal{Y}$ be two MSPs. We say that {\em $\mathcal{X}$ precedes $\mathcal{Y}$} (notation $\mathcal{X} \rightsquigarrow \mathcal{Y}$) iff either
\begin{itemize}
\item there is an st-path $(e_1, \dots, e_l)$ and indices $1 \leq i<j \leq l$ such that $e_i \in \mathcal{X}$ and $e_j \in \mathcal{Y}$, or
\item $\mathcal{X} \cap \mathcal{Y} \neq \emptyset$ and $\mathcal{X}$ is a BHSP and $\mathcal{Y}$ is a FHSP.
\end{itemize}
More specifically, we say that {\em $\mathcal{X}$ immediately precedes $\mathcal{Y}$} 
(notation $\mathcal{X} \overset{i}\rightsquigarrow \mathcal{Y}$)
iff $\mathcal{X} \rightsquigarrow \mathcal{Y}$ and either
\begin{itemize}
\item there is an st-path $(e_1, \dots, e_l)$ and indices $1 \leq i<j \leq l$ such that $e_i \in \mathcal{X}$ and $e_j \in \mathcal{Y}$
and $\forall i<k<j$ there does not exist any third MSP $\mathcal{Z}$ such that $e_k \in \mathcal{Z}$, or
\item $\mathcal{X} \cap \mathcal{Y} \neq \emptyset$ and $\mathcal{X}$ is a BHSP and $\mathcal{Y}$ is a FHSP (see second bullet above).
\end{itemize}
\end{definition}

Figure~\ref{fig:nondisjointSPs} illustrates the second bullet of Definition~\ref{def:sp_prec}, i.e.\ the special case of precedence of two non-disjoint MSPs $\mathcal{X}$ and $\mathcal{Y}$ where $\mathcal{X} \overset{i}\rightsquigarrow \mathcal{Y}$.
\begin{figure}
    \centering
    \includegraphics[width=0.2\linewidth]{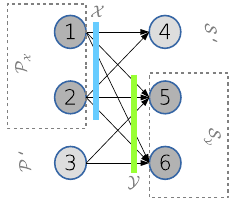}
    \caption{An example of two non-disjoint MSPs, where $\mathcal{X} \overset{i}\rightsquigarrow \mathcal{Y}$}
    \label{fig:nondisjointSPs}
\end{figure}

\begin{definition}
\label{def:msp_dag}
For a given st-DAG $G=(V,E)$, the {\em maximum syncpoint DAG (MSP-DAG) $G^M$} is defined as $G^M=(V^M,E^M)$ where $V^M$ ist the set of MSPs of $G$ and $E^M$ is given by the immediate precedence relation $\overset{i}\rightsquigarrow$ of Definition~\ref{def:sp_prec}.
\end{definition}

As opposed to the original st-DAG $G$, which according to Section~\ref{sec:defs_and_props} should not contain redundant edges, the maximum syncpoint DAG $G^M$ may contain redundant edges. These edges must not be deleted, since they are essential for the cluster finding algorithm in Section~\ref{sec:clustfind}.

\begin{lemma}
The MSP-DAG defined in Definition~\ref{def:msp_dag} has a single source MSP and a single target MSP (i.e.\ it is also an st-DAG). Both of them are FSPs.
\end{lemma}

\begin{proof}
This follows from Lemma~\ref{lem:sp_front_end}.
Given an st-DAG $G=(V,E)$, the set $\mathcal{X}$ of edges leaving the source vertex $s$ is is a FSP, which is the source MSP of the MSP-DAG. A symmetric argument can be given for the set of edges in $G$ leading to the target vertex $t$.
\end{proof}

Figure~\ref{fig:DAGandMSPDAG1} shows a DAG and the associated MSP-DAG. The latter contains three redundant edges
(namely
$\mathcal{X} \overset{i}\rightsquigarrow \mathcal{Y}$,
$\mathcal{B} \overset{i}\rightsquigarrow \mathcal{E}$ and
$\mathcal{A} \overset{i}\rightsquigarrow \mathcal{H}$),
which will be essential for our cluster finding algorithm.
In particular, the (redundant) edge from FSP $\mathcal{A}$ to FSP $\mathcal{H}$ will enable the algorithm to find the minmal complex cluster induced by vertices $\{ 9, \dots, 20 \}$. 

\begin{figure}
    \begin{minipage}[c]{0.5\linewidth}
    \centering
    \includegraphics[width=0.9\linewidth]{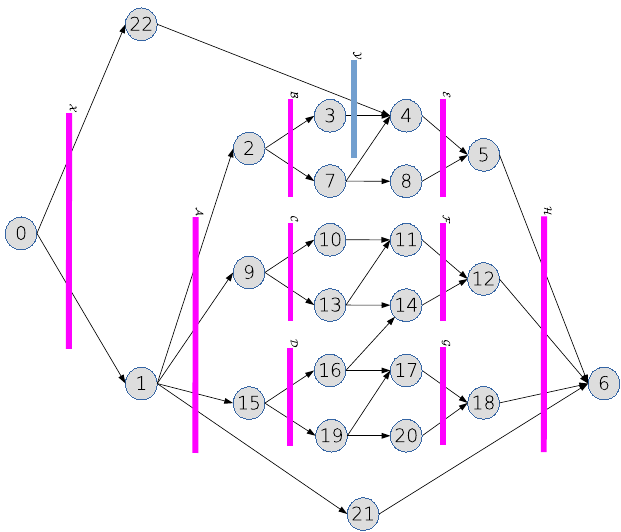}
    \end{minipage}
    \begin{minipage}[c]{0.5\linewidth}
    \centering
    \includegraphics[width=0.9\linewidth]{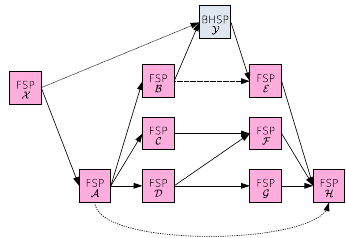}
    \end{minipage}
    \caption{A DAG (left) and the associated MSP-DAG (right)}
    \label{fig:DAGandMSPDAG1}
\end{figure}

Another example of a DAG and its associated MSP-DAG is shown in Fig.~\ref{fig:DAGandMSPDAG2}. Again, the MSP-DAG contains redundant edges
(namely
$\mathcal{A} \overset{i}\rightsquigarrow \mathcal{E}$,
$\mathcal{C} \overset{i}\rightsquigarrow \mathcal{E}$ and
$\mathcal{D} \overset{i}\rightsquigarrow \mathcal{F}$).
However, in this case these redundant edges do not connect any pairs of MSPs which would give rise to a minimal cluster (notice that the subgraph induced by vertex set $\{ 2,3,4,5,8,9,10,11\}$, enclosed by SSPs of MSPs $\mathcal{A}$ and $\mathcal{E}$, is a cluster, but not a minimal cluster).
In fact, the only minimal cluster in this DAG is induced by the vertex set $\{4,5,8,9 \}$, whose brackets are MSP $\mathcal{B}$ and a SSP of MSP $\mathcal{E}$.

\begin{figure}[h]
    \begin{minipage}[c]{0.5\linewidth}
    \centering
    \includegraphics[width=0.9\linewidth]{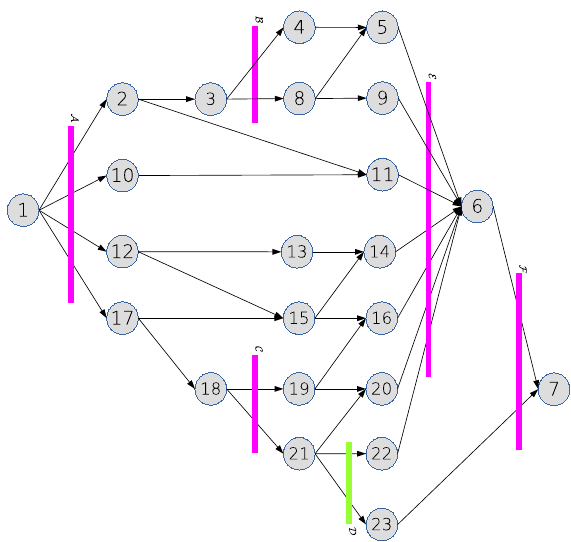}
    \end{minipage}
    \begin{minipage}[c]{0.5\linewidth}
    \centering
    \includegraphics[width=0.9\linewidth]{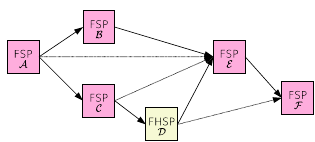}
    \end{minipage}
    \caption{Another DAG (left) and the associated MSP-DAG (right)}
    \label{fig:DAGandMSPDAG2}
\end{figure}

\section{Minimal Cluster Finding Algorithm}
\label{sec:clustfind}

\subsection{Key concepts and overview of the algorithm}
\label{subsec:key_conc}

A cluster is always framed by two syncpoints (see Lemma~\ref{lem:clust_SP}). More precisely, it is framed by an opening FSP/FHSP $\mathcal{X}$ and a closing FSP/BHSP $\mathcal{Y}$ (where $\mathcal{X}$ and $\mathcal{Y}$ do not have to be maximum syncpoints, they may also be subsyncpoints).
Therefore, our algorithm $FindAllMinClusters$ systematically checks candidates of possible cluster-enclosing MSP pairs $(\mathcal{X},\mathcal{Y})$.
These candidate bracket pairs are identified
by inspecting the MSP-DAG.
Since we are only interested in {\em minimal} clusters, the algorithm processes candidate pairs in order of increasing distance between $\mathcal{X}$ and $\mathcal{Y}$. In order to find all relevant MSP pairs, the algorithm -- in a preprocessing step -- constructs all possible relevant MSP-paths from some FSP/FHSP $\mathcal{X}$ to a FSP/BHSP $\mathcal{Y}$ (in the MSP-DAG).
These MSP-paths are stored in lists $PathsTo[\mathcal{Y}]$ (one list per MSP $\mathcal{Y}$), each one sorted in sections by increasing path length. For details on how to construct all of these relevant MSP-paths see \ref{sec:computingpaths}.

Algorithm~\ref{alg:findallminclust} shows the top-level algorithm for finding all minimal clusters in an st-DAG. It takes two arguments: The st-DAG $G$ and its associated MSP-DAG $G^M$.
It calls the function $ClusterCheck$ shown in Algorithm~\ref{alg:clustercheck}, which takes as arguments (apart from $G$) a pair of MSPs $\mathcal{X}$ and $\mathcal{Y}$ to be checked for in-between clusters.
$ClusterCheck$, in turn, calls the function $ClosedCheck$, shown in Algorithm~\ref{alg:closedcheck} whose job is more specific:
For two candidate vertex sets $X\!set$ and  $Y\!set$ (where $X\!set$ is derived from $\mathcal{X}$,
and $Y\!set$ is derived from $\mathcal{Y}$), it checks whether they are indeed the sets of entry resp.\ exit vertices of a cluster.

\begin{algorithm}[ht]
    \caption{$FindAllMinClusters$ finds all minimal clusters in an st-DAG $G=(V,E)$ whose  MSP-DAG ist given by $G^M=(V^M,E^M)$}\label{alg:findallminclust}
    \begin{algorithmic}[1]
        \State {\bf function} $FindAllMinClusters$ ($G$: st-DAG, $G^M$: MSP-DAG) {\bf returns} void
        \For {all FSPs and BHSPs $\mathcal{Y} \in V^M$}
        \State construct the list $PathsTo[\mathcal{Y}]$ of all MSP-paths in $G^M$ of any length starting at any FSP or FHSP $\mathcal{X} \in V^M$ and ending in $\mathcal{Y}$
        \EndFor
        \State /// the lists $PathsTo[\mathcal{Y}]$ contain all relevant paths to be checked for clusters
        \State /// each list $PathsTo[\mathcal{Y}]$ is already sorted (in sections) by construction
        \State /// for details on how to construct these lists efficiently, see \ref{sec:computingpaths}
        \For{all MSP-paths $p$ in any list $PathsTo[\mathcal{Y}]$ in order of incr.\ length}
        \label{line:FAMC7}
        \State mark $p$ as checked
        \State /// each MSP-path $p$ should only be checked once
        \State $\mathcal{X} := StartMSP(p)$
        \State $\mathcal{Y} := EndMSP(p)$
        \If{$|S_\mathcal{X}| \geq 2 \wedge |P_\mathcal{Y}| \geq 2 \wedge (S_\mathcal{X} \cap P_\mathcal{Y} = \emptyset)$}
        \State /// make sure that MSPs $\mathcal{X}$ and $\mathcal{Y}$ are big enough
        \State $found := ClusterCheck(G, \mathcal{X},\mathcal{Y})$
        \EndIf
        \If{$found$}
        \State delete $p$ and all of its superpaths
        \label{line:FAMC16}
        \State /// these superpaths can be found in the same list $PathsTo[\mathcal{Y}]$ and
        \State /// in all other lists $PathsTo[\mathcal{Z}]$ where MSP $\mathcal{Z}$ is reachable from $\mathcal{Y}$
        \State /// using the described pointer/separator structure
        \EndIf
        \EndFor
    \end{algorithmic}
\end{algorithm}

\subsection{Explanation of the algorithm}

Algorithm $FindAllMinClusters$ first performs the mentioned preprocessing, i.e.\ the construction of all lists $PathsTo[.]$.
Then, in the for-loop of line~\ref{line:FAMC7}, it goes through all relevant MSP-path candidates (in some order of increasing path length) whose start- and end-MSP (denoted $\mathcal{X}$ and $\mathcal{Y}$ in the algorithm) might enclose one or more clusters.
MSP-paths which have been investigated are marked as $checked$, and in case a cluster has been found while looking at path $p$, that $p$ and all of its superpaths are deleted (since we are only interested in minimal clusters, superpaths of successful paths do not have to be investigated).

The checking of a particular MSP pair $(\mathcal{X},\mathcal{Y})$ for in-between clusters is done by function $ClusterCheck$ (see Algorithm~\ref{alg:clustercheck}). Notice that there may be even more than one cluster enclosed by $\mathcal{X}$ and $\mathcal{Y}$. Correspondingly, we must realize that, in general, not all vertices of $S_{\mathcal{X}}$ have to belong to the cluster(s) (recall from Definition~\ref{def:syncpoint} that for an MSP $\mathcal{X}$, we denote by $P_\mathcal{X}$ the set of its start/predecessor vertices and by  $S_\mathcal{X}$ the set of its end/successor vertices). Only a subset of $S_{\mathcal{X}}$ may belong to some cluster (as its entry nodes), and a symmetrical statement can be made about the set $P_{\mathcal{Y}}$ of the closing MSP $\mathcal{Y}$.
Therefore, the function $ClusterCheck$ matches subsets $X\!set$ of $S_{\mathcal{X}}$ with subsets $Y\!set$ of $P_{\mathcal{Y}}$ which are mutually reachable in forward/backward direction.
The subgraph between $X\!set$ and $Y\!set$ is a candidate for a cluster, but it must be further checked by the function $ClosedCheck$ (see Algorithm~\ref{alg:closedcheck}).
Function $ClosedCheck$ builds the intersection of the forward reachability set of $X\!set$ and the backward reachability set of $Y\!set$. The induced subgraph resulting from this intersection is then checked for closedness, i.e.\ it must not have any entry vertices apart from $X\!set$ and not any exit vertices apart from $Y\!set$. If these conditions are satisfied, then the subgraph between $X\!set$ and $Y\!set$ is indeed a cluster.

\begin{algorithm}[ht]
\caption{$ClusterCheck$ checks for MSPs $\mathcal{X}$ and $\mathcal{Y}$ (where $\mathcal{X} \rightsquigarrow \mathcal{Y}$) whether there are pairs of a subset of $S_\mathcal{X}$ and a subset of $P_\mathcal{Y}$ such that the graph between them is a minimal cluster}\label{alg:clustercheck}  
\begin{algorithmic}[1]
        \State {\bf function} $ClusterCheck$ ($G$: st-DAG, $\mathcal{X},\mathcal{Y}$: MSP) {\bf returns} boolean
        \State $Covered := \emptyset$
        \State /// $Covered$ will be the subset of $S_\mathcal{X}$ including all vertices which have \State /// been processed in the following forward-backward reachability scheme
        \State $foundclusters := false$
        \While{$Covered \neq S_\mathcal{X}$}
        \State /// repeat this loop as long as there are vertices in $S_\mathcal{X}$ that have \State /// not been covered/tested
        \State $x :=$ some vertex of $S_\mathcal{X}$ that is not yet in set $Covered$
        \State $X\!set := \{x\}$
        \State $Y\!set := \emptyset$
        \State $stable := false$
        \While{not $stable$}
        \State /// this loop realizes a forward-backward reachability scheme
        \State $Y\!setnew :=$ all vertices of $P_\mathcal{Y}$ forward reachable from $X\!set$
        \If{$Y\!setnew \neq \emptyset$}
        \State /// $Y\!setnew == \emptyset$ would mean that MSP $\mathcal{Y}$ is not reachable
        \State /// from vertex $x$, in which case we can stop the
        \State /// forward-backward scheme
        \State $X\!setnew :=$ all vertices of $S_\mathcal{X}$ backward reachable from $Y\!setnew$
        
\Else
         \State $X\!setnew := X\!set$
\EndIf
\State $stable := (X\!setnew == X\!set \wedge Y\!setnew == Y\!set)$
\State $X\!set := X\!setnew$
\State $Y\!set := Y\!setnew$
\EndWhile
\State $Covered := Covered \cup X\!set$
\If{$|X\!set| \geq 2 \wedge |Y\!set| \geq 2$}
\If{$ClosedCheck(G, X\!set,Y\!set)$}
\State output subgraph between $X\!set$ and $Y\!set$ as a minimal cluster
\State $foundclusters := true$
\EndIf
\EndIf
\EndWhile
\State return $foundclusters$
\end{algorithmic}
\end{algorithm}

\begin{algorithm}[ht]    \caption{$ClosedCheck$ checks whether the subgraph between $X\!set$ and $Y\!set$ is closed, i.e. can be entered only via $X\!set$ and exited only via $Y\!set$}\label{alg:closedcheck}  
\begin{algorithmic}[1]
\State {\bf function} $ClosedCheck$ ($G$: st-DAG, $X\!set,Y\!set$: vertexset) {\bf returns} boolean
\State $VertexSet := FwReach(X\!set) \cap BwReach(Y\!set)$
\If{$VertexSet == \emptyset$}
\State return $false$
\EndIf
\State /// this if-clause is actually redundant, because the calling function
\State /// already ensured that $Y\!set$ is reachable from $X\!set$
\State $NonEntries := VertexSet \setminus X\!set$
\If{$\exists v \in NonEntries$ which has predecessor vertex not in $VertexSet$}
\label{line:ClosedCheck8}
\State return $false$
\EndIf
\State $NonExits := Vertices \setminus Y\!set$
\If{$\exists v \in NonExits$ which has successor vertex not in $VertexSet$}
\label{line:ClosedCheck11}
\State return $false$
\EndIf
\State return $true$
\end{algorithmic}
\end{algorithm}

\subsection{Deletion of MSP-superpaths}

As already mentioned, once a successful MSP-path $p$ has been found, all superpaths of $p$ can safely be deleted/ignored from all lists $PathsTo[.]$, since we are only interested in {\em minimal} clusters. This is done in line~\ref{line:FAMC16} of Algorithm~\ref{alg:findallminclust}.
A superpath of $p$ is an extension of $p$ at the front or/and at the tail. For example, the MSP-path $[\mathcal{A},\mathcal{B},\mathcal{C},\mathcal{D}]$ is a superpath of $[\mathcal{A},\mathcal{B},\mathcal{C}]$, but it is not a superpath of $[\mathcal{A},\mathcal{C},\mathcal{D}]$ (here we see that is important to keep all redundant edges in the MSP-DAG).
All superpaths of the current MSP-path $p$ can be found efficiently by setting up a pointer structure between the entries of the lists $PathsTo{[.]}$, and by adding sublist separators of the form $|_{\scriptscriptstyle n}$ within each of those lists, which all can be done during Algorithm~\ref{alg:computeallvpaths} at practically no additional cost.
More precisely, when Algorithm~\ref{alg:computeallvpaths} (in line~\ref{line:cavp11}) adds the new path $p_{new}$ to the list $PathsTo[v]$, a pointer from $p$ in $PathsTo[u]$ to $p_{new}$ in $PathsTo[v]$ is created.
Furthermore, while processing the various successors $v$ during the for-loop of line~\ref{line:cavp8}, the algorithm 
inserts a new sublist separator $|_{\scriptscriptstyle 1}$ at the end of the current list $PathsTo[v]$ before moving to the next successor $v$.
In addition, existing sublist separators are inherited from the list $PathsTo[u]$ while going through that list in the for-loop of line~\ref{line:cavp9}, but their index is increased by one.
Consequently, the index $n$ of a separator $|_{\scriptscriptstyle n}$ indicates how much the suffixes of paths coincide across the separator.
In other words, the index $n$ of a separator $|_{\scriptscriptstyle n}$ indicates up to which length of an originating path the separator should be ignored.
Therefore, starting with a path whose vertex length (denoted $vlen$) is $l$, and moving along within the same list, all separators $|_{\scriptscriptstyle n}$ with $l \leq n$ are ignored.
With the help of these pointers and sublist separators, superpaths arising from additional prefixes as well as postfixes can be found easily by moving along the same list up to the next (valid) separator and by following the pointers into the other lists.
The described pointer/separator structure thus allows for the efficient deletion of all superpaths of a successful path.
Notice that the creation of the pointers/separators is not part of the code in Algorithm~\ref{alg:computeallvpaths}, in order to keep the algorithm clear and simple.

For an example of the pointer/separator structure just described see Figure~\ref{fig:PathsTo}, which corresponds to the MSP-DAG of Figure~\ref{fig:DAGandMSPDAG2}.
The figure shows all lists $PathsTo[.]$, one list per line, as resulting from Algorithm~\ref{alg:computeallvpaths}.
The figure shows the pointer structure only for two example paths: For the path
$[\mathcal{B}\mathcal{E}]$ (in blue) and for the path
$[\mathcal{D}\mathcal{E}]$ (in red).
Within each list, the sublist separators are shown by the $|_{\scriptscriptstyle n}$ symbol.
Notice that, starting with the path $[\mathcal{E} \mathcal{F}]$ whose vertex length is
$vlen([\mathcal{E} \mathcal{F}])=2$
and while moving along within the same list, all separators $|_{\scriptscriptstyle 2}$ are considered `invalid' and therefore ignored.
Notice that for this example from Figure~\ref{fig:DAGandMSPDAG2}, there is only one MSP pair actually leading to a minimal cluster, which is the pair $(\mathcal{B},\mathcal{E})$ corresponding to the `successful' path $[\mathcal{B}\mathcal{E}]$.
That means that only the blue pointers in Figure~\ref{fig:PathsTo}
would actually be followed while deleting superpaths of a successful path.

\begin{figure}
\centering
\scalebox{0.65}{
\begin{tikzpicture}
\node (rootA) at (-13.5,-0.6) {$PathsTo[\mathcal{A}]:$};
\node (A) at (-12,-0.6) {$[[\mathcal{A}]]$};
\node (rootB) at (-13.5,-1.4) {$PathsTo[\mathcal{B}]:$};
\node (B) at (-12,-1.4) {$[[\mathcal{B}]$};\node [right = -0.1cm of B] (AB) {$[\mathcal{AB}]]$};
\node (rootC) at (-13.5,-2.2) {$PathsTo[\mathcal{C}]:$};
\node (C) at (-12,-2.2) {$[[\mathcal{C}]$};\node [right = -0.1cm of C] (AC) {$[\mathcal{AC}]]$};
\node (rootD) at (-13.5,-3) {$PathsTo[\mathcal{D}]:$};
\node (D) at (-12,-3) {$[[\mathcal{D}]$};\node [right = -0.1cm of D] (CD) {$[\mathcal{CD}]$};
\node [right = -0.1cm of CD] (ACD) {$[\mathcal{ACD}]]$};
\node (rootE) at (-13.5,-3.8) {$PathsTo[\mathcal{E}]:$};
\node (E) at (-12,-3.8) {$[[\mathcal{E}]$};\node [right = -0.1cm of E] (AE) {$[\mathcal{AE}]$};
\node [right = -0.1cm of AE] (sep1) {$\Big|_1$};
\node [right = -0.2cm of sep1] (BE) {$[\mathcal{BE}]$};
\node [right = -0.1cm of BE] (ABE) {$[\mathcal{ABE}]$};
\node [right = -0.1cm of ABE] (sep2) {$\Big|_1$};
\node [right = -0.2cm of sep2] (CE) {$[\mathcal{CE}]$};
\node [right = -0.1cm of CE] (ACE) {$[\mathcal{ACE}]$};
\node [right = -0.1cm of ACE] (sep3) {$\Big|_1$};
\node [right = -0.2cm of sep3] (DE) {$[\mathcal{DE}]$};
\node [right = -0.1cm of DE] (CDE) {$[\mathcal{CDE}]$};
\node [right = -0.1cm of CDE] (ACDE) {$[\mathcal{ACDE}]]$};
\node (rootF) at (-13.5,-4.8) {$PathsTo[\mathcal{F}]:$};
\node (F) at (-12,-4.8) {$[[\mathcal{F}]$};\node [right = -0.1cm of F] (DF) {$[\mathcal{DF}]$};
\node [right = -0.1cm of DF] (CDF) {$[\mathcal{CDF}]$};
\node [right = -0.1cm of CDF] (ACDF) {$[\mathcal{ACDF}]$};
\node [right = -0.1cm of ACDF] (sep4) {$\Big|_1$};
\node [right = -0.2cm of sep4] (EF) {$[\mathcal{EF}]$};
\node [right = -0.1cm of EF] (AEF) {$[\mathcal{AEF}]$};
\node [right = -0.1cm of AEF] (sep5) {$\Big|_2$};
\node [right = -0.2cm of sep5] (BEF) {$[\mathcal{BEF}]$};
\node [right = -0.1cm of BEF] (ABEF) {$[\mathcal{ABEF}]$};
\node [right = -0.1cm of ABEF] (sep6) {$\Big|_2$};
\node [right = -0.2cm of sep6] (CEF) {$[\mathcal{CEF}]$};
\node [right = -0.1cm of CEF] (ACEF) {$[\mathcal{ACEF}]$};
\node [right = -0.1cm of ACEF] (sep7) {$\Big|_2$};
\node [right = -0.2cm of sep7] (DEF) {$[\mathcal{DEF}]$};
\node [right = -0.1cm of DEF] (CDEF) {$[\mathcal{CDEF}]$};
\node [right = -0.1cm of CDEF] (ACDEF) {$[\mathcal{ACDEF}]]$};
\draw [dashed, blue,->,out=80,in=100,looseness=0.8] ([xshift=0mm] BE.north) to ([xshift=0mm] ABE.north);
\draw [blue,->,out=-80,in=150,looseness=0.3] ([yshift=0mm] BE.south) to ([xshift=-1mm] BEF.north);
\draw [dashed, blue,->,out=80,in=100,looseness=0.4] ([xshift=1mm] BEF.north) to ([xshift=0mm] ABEF.north);
\draw [dashed, green,->,out=-80,in=-100,looseness=0.6] ([xshift=0mm] EF.south) to ([xshift=-1mm] AEF.south);
\draw [dashed, green,->,out=-80,in=-100,looseness=0.6] ([xshift=1mm] AEF.south) to ([xshift=-1mm] BEF.south);
\draw [dashed, green,->,out=-80,in=-100,looseness=0.6] ([xshift=1mm] BEF.south) to ([xshift=-1mm] ABEF.south);
\draw [dashed, green,->,out=-80,in=-100,looseness=0.6] ([xshift=1mm] ABEF.south) to ([xshift=-1mm] CEF.south);
\draw [dashed, green,->,out=-80,in=-100,looseness=0.6] ([xshift=1mm] CEF.south) to ([xshift=-1mm] ACEF.south);
\draw [dashed, green,->,out=-80,in=-100,looseness=0.6] ([xshift=1mm] ACEF.south) to ([xshift=-1mm] DEF.south);
\draw [dashed, green,->,out=-80,in=-100,looseness=0.6] ([xshift=1mm] DEF.south) to ([xshift=-1mm] CDEF.south);
\draw [dashed, green,->,out=-80,in=-100,looseness=0.6] ([xshift=1mm] CDEF.south) to ([xshift=0mm] ACDEF.south);
\draw [dashed, red,->,out=80,in=100,looseness=1] ([xshift=0mm] DE.north) to ([xshift=-1mm] CDE.north);
\draw [dashed, red,->,out=80,in=100,looseness=0.8] ([xshift=1mm] CDE.north) to ([xshift=0mm] ACDE.north);
\draw [red,->,out=-50,in=100,looseness=0.3] ([yshift=0mm] DE.south) to ([xshift=-1mm] 
DEF.north);
\draw [dashed, red,->,out=80,in=100,looseness=1] ([xshift=1mm] DEF.north) to ([xshift=-1mm] CDEF.north);
\draw [dashed, red,->,out=80,in=100,looseness=0.7] ([xshift=1mm] CDEF.north) to ([xshift=0mm] ACDEF.north);
\end{tikzpicture}
}
\caption{The lists $PathsTo[.]$ corresponding to the MSP-DAG of Figure~\ref{fig:DAGandMSPDAG2}; 
the sublist separators are shown as the $|_{\scriptscriptstyle n}$ symbol;
only part of the pointer structure is shown, namely for path $[\mathcal{BE}]$ (in blue), for path $[\mathcal{DE}]$ (in red) and for path $[\mathcal{EF}]$ (in green);
notice that separators $|_{\scriptscriptstyle 2}$ are ignored while moving along the list from $[\mathcal{EF}]$ because $vlen([\mathcal{EF}]) \leq 2$ (where $vlen$ counts the number of vertices of a path);
the dashed pointers within the same list are only shown for illustrative purposes, they are not actually implemented because the sublist separators suffice to indicate where the sublist of superpaths ends.}
\label{fig:PathsTo}
\end{figure}

\subsection{Complexity analysis}
\label{subsec:complexity}

In this section, we analyze the runtime complexity of the presented minimal cluster finding algorithm.

\begin{theorem}
\label{thm:complexity}
Given an st-DAG $G=(V,E)$ and the associated MSP-DAG $G^M=(V^M,E^M)$, with $n = |V|$ and 
$N=|V^M|$.
Let the maximum number of start/end vertices of any maximum syncpoint be
$k := \max \{ \max_{\mathcal{X} \in V^M} |S_{\mathcal{X}}|, \max_{\mathcal{Y} \in V^M} |P_\mathcal{Y}| \}$.
The overall runtime complexity of Algorithm~\ref{alg:findallminclust} for finding all minimal clusters in $G$ is
\begin{equation*}
O(k\cdot n^2 \cdot 2^N)
\end{equation*}
\end{theorem}

Essentially, the theorem states that the complexity is quadratic in the number of vertices of $G$ and exponential in the number of vertices of $G^M$, the latter due to the worst case number of relevant paths in $G^M$ that need to be checked.

\begin{proof}
For a given st-DAG $G=(V,E)$, let $n:=|V|,m := |E|$. Likewise, for the associated MSP-DAG $G^M=(V^M,E^M)$, let $N:=|V^M|,M := |E^M|$.

The following information, needed during different stages of the cluster finding process, can be precomputed:
A topological order of $V$ (of $V^M$) can be computed, e.g.\ by Kahn's algorithm \citep{kahn:1962}, in time $O(n+m)$ (in time $O(N+M)$).
The sets of all vertices (forward) reachable from any given vertex $v \in V$, denoted by $FwReach(v)$, can be computed and output
in a single traversal of $G$ (in reverse topological order) in time $O(n^2)$
(for all $v$ at once!).
Symmetrically, all sets of vertices from which any vertex $v$ is reachable, denoted by $BwReach(v)$, can be computed and output by a single sweep (in topological order) in time $O(n^2)$.
During these two traversals, the immediate successor sets $Succ(v)$ and the immediate predecessor sets $Pred(v)$ of all vertices $v \in V$ can also be obtained and stored at no additional cost.
Therefore the overall cost of precomputations is $O(n^2)$.
Notice that the sets $Pred(v)$, $Succ(v)$, $FwReach(v)$ and $BwReach(v)$ can be stored compactly with the help of symbolic data structures (e.g.\ BDD \citep{bryant:1986}), and that set operations on them can be performed very efficiently.

The overall number of elementary operations (ops) needed to execute algorithm $FindAllMinClusters$ is:
\begin{eqnarray}
\label{eqn:complex}
&ops(compute \; all \; relevant \; paths) \; + & \\
&NumberO\!f\!RelevantPaths \cdot (ops(ClusterCheck)+ops(DeleteSuperpaths)) \nonumber&
\end{eqnarray}
where the cost for $ClosedCheck$ will have to be included in that for $ClusterCheck$.

The computation of all relevant paths of $G^M$ has time complexity $O(N \cdot 2^N)$ (see \ref{sec:computingpaths}).

The time complexity of algorithm $ClusterCheck$ is dominated by the forward-backward reachability scheme between the vertices of $S_\mathcal{X}$ and $P_\mathcal{Y}$ and by the calls to function $ClosedCheck$.
MSPs may have different sizes, therefore we define $k := \max \{ \max_{\mathcal{X} \in V^M} |S_{\mathcal{X}}|, \max_{\mathcal{Y} \in V^M} |P_\mathcal{Y}| \}$, where we typically have $k \ll n$.
In the worst case, the number of steps of the forward-backward reachability scheme
(in total, i.e.\ accounting for {\em all} iterations of the outer while-loop of $ClusterCheck$)
is $2k-1$ (i.e.\ $k$ forward steps and $k-1$ backward steps).
Since each of these steps involves an intersection operation with $P_\mathcal{Y}$ resp.\ $S_\mathcal{X}$, each having a maximum of $k$ elements,
the number of operations needed to complete the forward-backward reachability scheme is $O(k^2)$ (remember that all $FwReach(v)$ and $BwReach(v)$ have already been precomputed).
The maximum number of calls of function $ClosedCheck$ during one call of $ClusterCheck$ is  $O(k)$, and each of those calls needs at most $O(n^2)$ operations.
To see this, notice that the runtime of function $ClosedCheck$ is dominated by line~\ref{line:ClosedCheck8} (and similarly line~\ref{line:ClosedCheck11}).
In line~\ref{line:ClosedCheck8}, the union of predecessors of all vertices of the set $NonEntries$ needs to be computed, in order to be intersected with the complement of $VertexSet$.
Although the individual predecessor sets have already been precomputed, it takes $O(n^2)$ operations to build their union.
So the number of steps for all calls to $ClosedCheck$ together is $O(k \cdot n^2)$.
Therefore the overall time needed per call of $ClusterCheck$ is $O(k^2 +k\cdot n^2) = O(k  \cdot n^2)$.

For a particular path $p$ which was found to be a successful MSP-path, in the worst case up to $2^N-(n+2)$ superpaths need to be deleted (in the extreme case all other paths apart from $p$). However, during the the for-loop of $FindAllMinimalClusters$, every path will be deleted at most once, so the total number of MSP-path deletions is $O(2^N)$. This means that the complexity of MSP-path deletions is already taken into account by the factor $NumberO\!f\!RelevantPaths$ which is $2^N$.

Putting together these partial complexity results according to Equation~\ref{eqn:complex}, and noticing that $k \cdot n^2 > N$, the overall runtime complexity of $FindAllMinClusters$ becomes:
\begin{equation}
O(N \cdot 2^N) + 
O(2^N \cdot ((k \cdot n^2)+0))
= O(k \cdot n^2 \cdot 2^N)
\end{equation}
To this we have to add the cost $O(n^2)$ for the precomputations (see beginning of this subsection), such that the overall complexity of Algorithm~\ref{alg:findallminclust} becomes $O(n^2+k\cdot n^2 \cdot 2^N) = O(k\cdot n^2 \cdot 2^N)$.
\end{proof}

\section{Applications}
\label{sec:appl}

This section presents some experimental results, testing our minimal cluster finding algorithm on a wide range of st-DAGs and thereby collecting some interesting statistics.

\subsection{Generation of random DAGs}
\label{subsec:randgen}
In order to generate test cases for our minimal cluster finding algorithm, we use a random DAG generation scheme.
It starts with a 3-vertex DAG which gets expanded step by step until the desired number of vertices is reached.
The generation is steered by the following parameters:
\begin{itemize}
\item 
$n$: total number of vertices of the final DAG
\item 
$parexp \geq 0$: probability of parallel expansions
\item 
$serexp \geq 0$: probability of serial expansions
\item 
$maxwidth \in \mathbb{N}$: maximum width for parallel expansions resp.\ subgraph expansions
\item 
$clustsettle$: factor controlling the number of edges during subgraph expansion
\item 
$narb$: number of arbitrary (disruptive) edges to be added at the end
\end{itemize}
with the following constraints: $parexp + serexp \leq 1$ and $0 < clustsettle < 1.$

The algorithm starts with a DAG consisting of three vertices in series. The source vertex is $s$, the target vertex is $t$.
Then the following iteration steps are repeated until the desired number of vertices $n$ has been reached:
One vertex $x$ (other than $s$ or $t$) is chosen randomly to be expanded by one of the following three actions:
(1) With probability $parexp$, $x$ is replaced by $2 \leq k \leq maxwidth$ new vertices linked in parallel, where $k$ is chosen uniformly (remember that linked in parallel means that all new vertices have the same predecessors and successors which $x$ had).
(2) With probability $serexp$, $x$ is replaced by 2 serially linked vertices.
(3) With probability $1-parexp-serexp$, $x$ is replaced by a new subgraph of $k+l$ vertices where $k,l $ are chosen uniformly at random with $2 \leq k \leq maxwidth$ and $2 \leq l \leq maxwidth$. The subgraph has $k$ entry vertices and $l$ exit vertices,
where each entry vertex has an incoming edge from the predecessor of $x$ and each exit vertex has an outgoing edge to the successor of $x$.
Each entry vertex is connected to at least one exit vertex (chosen randomly), and each exit vertex can be reached by an edge from at least one entry vertex. Additional edges between the entry and exit vertices are added at random, such that the total number of edges of the new subgraph is larger than $max(k,l)+1$ but smaller than $k \cdot l \cdot clustsettle$.

At the end, once a DAG with $n$ vertices has been created by the above steps, $narb$ so-called arbitrary (disruptive) edges are added by randomly choosing a start and an end vertex, where it has to be ensured that the new disruptive edge does not lead to a cycle, that it is not redundant and that it does not make other already existing edges redundant. Therefore, when trying to add the new disruptive edge $(u,v)$, the following three conditions need to be checked:
(1) There must not be an existing path from $u$ to $v$.
(2) There must not be an existing edge from any ancestor $a$ of $u$ to $v$ (otherwise that edge $(a,v)$ would be made redundant).
(3) There must not be an existing edge from $u$ to any descendant $d$ of $v$ (otherwise that edge $(u,d)$ would be made redundant).

Notice that overshooting is possible, i.e.\ the resulting DAG might have slightly more vertices than the given $n$. This happens, if as the last expansion step parallel expansion or subgraph expansion is chosen and the number of added vertices overshoots $n$.
To be precise, in theory, the maximum possible number of vertices created is $n-1+2  \cdot maxwidth$, which occurs if the last step, starting from $n-1$ vertices, is subgraph expansion with the maximum possible values for $k$ and $l$.

An example DAG created by this scheme is shown in Figure~\ref{fig:exampleRandomDAG}.

\begin{figure}
    \centering
    \includegraphics[width=1.0\linewidth]{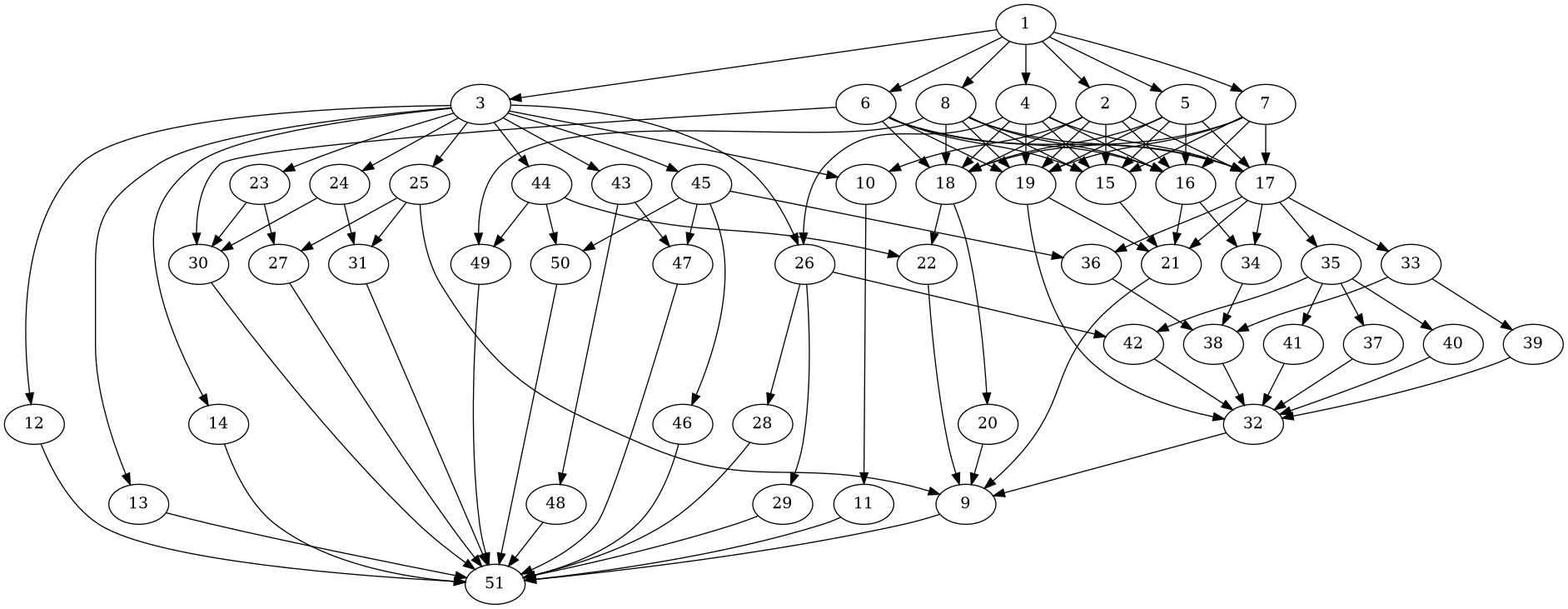}
    \caption{A randomly generated DAG. The generation parameters were
    $n=50$,
    $parexp = 0.33$,
    $serexp = 0.33$,
    $maxwidth = 5$,
    $clustsettle = 0.2$,
    $narb = 10$.
    % seed was 53145
    The 10 disruptive edges are 
    (44,22),
    (45,36),
    (25,9),
    (8,49),
    (6,30),
    (19,32),
    (25,27),
    (26,42),
    (16,34) and
    (4,26), added in that order.
    Notice the slight overshooting (resulting in 51 instead of 50 vertices).}
    \label{fig:exampleRandomDAG}
\end{figure}

\subsection{Experimental results}
\label{subsec:expresults}

In this section, we report on experiments with a prototype implementation of Algorithm~\ref{alg:findallminclust}, realized as a  single-threaded non-optimized python3 code (using the NetworkX package).
All experiments were run on a machine with an
Intel\circledR $\,$ Core i5-6400 processor
(2.7 GHz),
with 16 GB memory and 6 MB Intel\circledR $\,$ Smart Cache,  running Ubuntu 24.04.
The results are shown in Table~\ref{tab:results1}.

The experiments were run on randomly generated DAGs of diffent size, employing the DAG generation scheme of Section~\ref{subsec:randgen}.
The parameter $maxwidth$ was increased with growing number of vertices $n$, and similarly for the parameter $narb$.
This choice was motivated by the assumption that larger graphs should contain wider and bigger subgraphs than smaller ones, and they should also contain more arbitrary (disruptive) edges.
The other generation parameters were constant over all DAG sizes, with values as given in the caption of Table~\ref{tab:results1}.
In particular, serial, parallel and subgraph expansion occur with equal probabilities, which seems to be a good compromise.
Every row of the table was obtained from analyzing 100 random DAG samples.
However, for $n=200, \dots, 1000$ we ran two sets of experiments for each $n$.
While the statistics are, of course, very similar for each such pair of rows, there are some noticable differences, especially for the maximum runtimes in column (19). These differences stem from outliers which naturally occur since the DAGs are randomly generated.

From columns (4) and (7) of the table, we can see that the number of MSPs  $N = |V^M|$ is much smaller than the number of st-DAG vertices $n$ (very roughly only 10\%), 
which confirms our claim that dealing with the much smaller MSP-DAG is a major source of efficiency of our minimal cluster finding algorithm.

Columns (8) through (11) contain information about the number of minimal clusters found. These numbers are quite small, but more experiments (see below, at the end of this section) showed that other random DAG generation parameters can lead to a larger number of clusters.
Notice that working with these parameters, leading to relatively small numbers of clusters, doesn't generally make life easier for our algorithm, because those cases do not allow for many deletions of superpaths in line~\ref{line:FAMC16} of Algorithm~\ref{alg:findallminclust}, while superpath deletion obviously saves runtime.

The statistics on cluster sizes (rows (12) through (15)) were obtained over all 100 samples of each row.
The cluster size obviously varies a lot.
In each row, at least one cluster of minimal size 4 occured, as well as at least one cluster with more than $n$ vertices (in most cases this is the cluster containing all vertices of the DAG except the source $s$ and the target $t$) -- the variance of the cluster size thus being extremely large. Remember that cluster sizes larger than $n$ are of course possible because of the described overshooting effect during DAG generation.

\begin{landscape}

\begin{table}
    \centering
    \tiny{
    \begin{tabular}{||ccc||cccc||cccc|cccc|cccc||} \hline \hline
        \multicolumn{3}{||c||}{DAG $G$ generation parameters} 
        & \multicolumn{4}{c||}{$N=|V^M|$: number of MSPs per DAG } & \multicolumn{4}{c|}{number of clusters per DAG} & \multicolumn{4}{c|}{cluster sizes (vertices)} & \multicolumn{4}{c||}{runtime of Algo.~\ref{alg:findallminclust} per DAG (sec)}\\
        $n$ & $maxwidth$  & $narb$ & mean & var &  min & max & mean & var &  min & max & mean & var & min & max & mean & var & min & max \\

    (1) & (2)  & (3) & (4) & (5) & (6) & (7) & (8) & (9) & (10) & (11) & (12) & (13) & (14) & (15) & (16) & (17) & (18) & (19) \\ \hline \hline

    200 & 10  & 20 & 28.30 & 18.67 & 19 & 38 & 1.83 & 0.90 & 1 & 5 & 44.01 & 4718 & 4 & 206 & 0.081 & 0.003 & 0.015 & 0.277 \\ \hline

    200 & 10  & 20 & 28.59 & 26.58 & 15 & 45 & 1.65 & 0.85 & 1 & 6 & 42.45 & 4631 & 4 & 211 & 0.086 & 0.005 & 0.015 & 0.611 \\ \hline \hline

    400 & 13  & 40 & 48.37 & 49.19 & 29 & 68 & 2.18 & 1.47 & 1 & 6 & 49.07 & 12570 & 4 & 407 & 0.400 & 0.133 & 0.054 & 1.715 \\ \hline

    400 & 13  & 40 & 47.36 & 49.09 & 33 & 66 & 1.91 & 1.80 & 1 & 7 & 58.27 & 15881 & 4 & 410 & 0.427 & 0.167 & 0.069 & 2.674 \\ \hline \hline

    600 & 16  & 60 & 60.21 & 72.59 & 38 & 78 & 2.17 & 1.74 & 1 & 7 & 73.31 & 30177 & 4 & 615 & 0.854 & 0.597 & 0.156 & 4.221 \\ \hline

    600 & 16  & 60 & 60.94 & 79.16 & 36 & 78 & 2.19 & 1.49 & 1 & 6 & 69.04 & 27254 & 4 & 624 & 0.904 & 0.815 & 0.141 & 4.914 \\ \hline \hline

      800 & 19  & 80 & 72.66 & 84.32 & 53 & 97 & 2.06 & 1.50 & 1 & 7 & 88.57 & 51077 & 4 & 811 & 1.834 & 6.640 & 0.337 & 18.693 \\ \hline

      800 & 19  & 80 & 73.83 & 96.67 & 52 & 95 & 2.52 & 2.05 & 1 & 6 & 69.63 & 41166 & 4 & 813 & 1.767 & 2.077 & 0.316 & 6.403 \\ \hline \hline

      1000 & 22  & 100 & 78.36 & 115.87 & 56 & 102 & 2.14 & 1.88 & 1 & 8 & 80.60 & 64465 & 4 & 1031 & 2.147 & 3.515 & 0.446 & 9.920 \\ \hline

      1000 & 22  & 100 & 80.18 & 157.49 & 46 & 110 & 2.44 & 2.61 & 1 & 11 & 100.23 & 75202 & 4 & 1012 & 2.466 & 6.796 & 0.388 & 14.349 \\ \hline \hline \hline

      1500 & 26  & 150 & 106.78 & 181.61 & 78 & 147 & 2.39 & 2.22 & 1 & 8 & 112.03 & 135605 & 4 & 1523 & 6.939 & 55.36 & 1.058 & 51.267 \\ \hline 

      2000 & 31  & 200 & 125.05 & 290.45 & 91 & 183 & 2.74 & 2.89 & 1 & 9 & 101.05 & 172638 & 4 & 2033 & 14.295 & 220.18 & 1.557 & 76.050 \\ \hline 

      2500 & 36  & 250 & 136.87 & 372.59 & 89 & 215 & 2.71 & 2.65 & 1 & 8 & 120.56 & 245996 & 4 & 2537 & 23.982 & 912.75 & 2.730 & 227.175 \\ \hline 

      3000 & 39  & 300 & 153.35 & 287.01 & 107 & 194 & 2.86 & 4.26 & 1 & 9 & 212.93 & 556848 & 4 & 3061 & 34.150 & 3245.19 & 2.985 & 405.332 \\ \hline

    3500 & 42  & 350 & 170.37 & 428.25 & 127 & 210 & 3.07 & 3.27 & 1 & 8 & 244.89 & 698949 & 4 & 3536 & 63.823 & 7579.72 & 4.479 & 587.152 \\ \hline

      4000 & 45 & 400 & 185.96 & 671.58 & 112 & 256 & 3.19 & 4.13 & 1 & 11 & 158.99 & 558597 & 4 & 4029 & 75.317 & 8050.69 & 7.714 & 383.237 \\ \hline

      4500 & 48 & 450 & 199.65 & 731.59 & 148 & 271 & 3.06 & 4.56 & 1 & 10 & 196.58 & 765315 & 4 & 4530 & 106.679 & 17423.67 & 8.940 & 917.959 \\ \hline
      
    5000 & 50  & 500 & 211.07 & 497.93 & 147 & 265 & 3.68 & 6.20 & 1 & 12 & 159.34 & 706631 & 4 & 5020 & 139.014 & 44355.85 & 10.951 & 1638.558 \\ \hline \hline

%      x00 & xx  & xx & a & b & c & d & e & f & g & h & i & j & k & l & m & n & o & p \\ \hline \hline
    \end{tabular}
    }
    \caption{Experimental results showing the number of MSPs, the number of minimal clusters, their sizes and the runtime of Algorithm~\ref{alg:findallminclust} (the minimal cluster finding algorithm) for different DAG sizes. The number of random DAG samples per row was 100. The following graph generation parameters were fixed for all experiments in this table: $parexp=0.33$, $serexp=0.33$ and $clustsettle=0.4$}
    \label{tab:results1}
\end{table}
\end{landscape}

In the last four columns of Table~\ref{tab:results1}, information on the runtime of Algorithm~\ref{alg:findallminclust} is provided.
The increase in runtime with growing DAG size $n$ is initially very moderate, with some non-dramatic outliers as observed, for example, for $n=800$.
The absolute values of the runtime (at most a few seconds for up to $n=1000$, even on this rather low-end machine) confirm that our algorithm is efficient and useful for practice.

Figure~\ref{fig:plot}
graphically displays some of the information of Table~\ref{tab:results1}.
The measured mean runtime of Algorithm~\ref{alg:findallminclust} is shown in blue (solid line), and the measured mean number of MSPs in red.
For the runtime, we also fitted a function $f(n)=c_1 n^2 + c_2 \, n^2 \, 2^{N(n)}$ (by the least-squares method), which does not match exactly the expression of Theorem~\ref{thm:complexity}, but instead adds a purely quadratic dependence on $n$.
The computed coefficients are
$c1 = 4.882E{-}6$
% c1 = 4.882124916641209e-6
and
$c2 = 1.964E{-}70$
% c2 = 1.964111922544133e-70
(feature scaling was used during fitting in order to avoid overflow caused by the term $2^{N(n)}$).
This fitting shows that, 
in the plotted range, the exponential term does not yet play any significant part, but the runtime growth is largely quadratic.
Thus, the exponential growth of the runtime with respect to $N$, as stated by Theorem~\ref{thm:complexity}, cannot be observed yet for DAG sizes of up to $n=5000$, which of course is a good message.

Looking at the blue solid curve in Figure~\ref{fig:plot}, there seems to be a disturbance at the data point for $n=3500$. This could already be observed in Table~\ref{tab:results1} and is caused by some outliers, obviously stemming from some very ''tough'' DAGs generated randomly for $n=3500$ (see column (19) in Table~\ref{tab:results1}).
To support this presumption, we ran another experiment for $n=3500$ and identical parameters as the ones used in Table~\ref{tab:results1} (again containing 100 samples). Indeed, this time we observed a mean runtime of only
44.749 sec (which would avoid the bend in the solid blue line of Figure~\ref{fig:plot}),
and the maximum runtime was also much lower, namely only 338.488 sec.%(instead of the previous 587.152 sec).

As mentioned, the parameters used in Table~\ref{tab:results1} and Figure~\ref{fig:plot} led to a relatively small number of clusters. In order to show that DAGs may contain a much larger number of clusters, we ran another experiment for $n=1000$ (again consisting of 100 samples).
This time, the other parameters for the randomly generated DAGs were as follows:
$parexp=0.4$,
$serexp=0.4$,
$maxwidth=15$,
$clustsettle=0.2$
and $narb=15$.
This experiment led to a mean number of clusters of $22.23$ (as opposed to the earlier value $2.44$) and maximum cluster size of $34$ (as opposed to $11$).
The mean runtime for this set of experiments was 3.951 sec (as opposed to 2.466 sec).

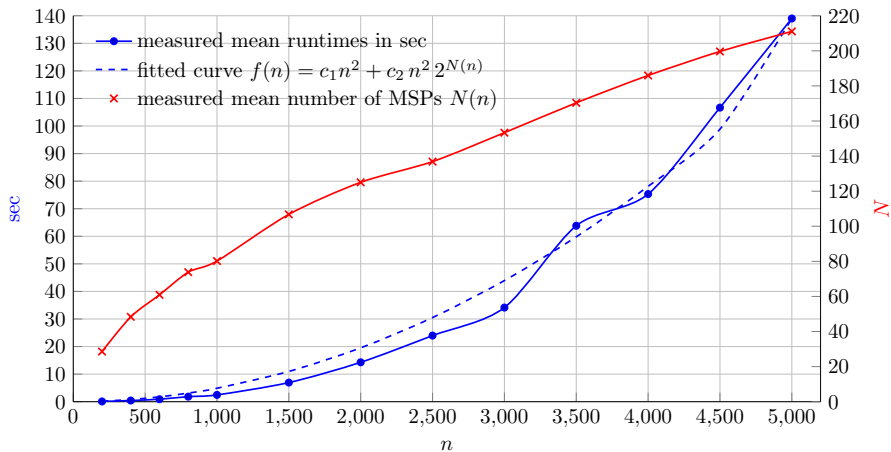
\begin{figure}[t]
\centering
\begin{tikzpicture}[scale=0.8, transform shape]

\begin{axis}[
  name=leftaxis,
  width=14cm,height=8cm,
  xmin=0,xmax=5200,
  xtick={0,500,...,5000},
  xlabel={$n$},
  ymin=0,ymax=140,
  ytick={0,10,...,140},
  ylabel={\color{blue} sec},
  axis y line*=left,
  axis x line*=bottom,
  grid=both,
  legend style={at={(0.02,0.98)},anchor=north west,draw=none,fill=none},
  legend cell align=left,
]

\addplot+[
  blue,
  mark=*,
  mark size=1.5pt,
  smooth,
  thick,
] coordinates {
  (200,0.086) (400,0.427) (600,0.904) (800,1.834) (1000,2.466)
  (1500,6.939) (2000,14.295) (2500,23.982) (3000,34.150)
  (3500,63.823) (4000,75.317) (4500,106.679) (5000,139.014)
};
\addlegendentry{measured mean runtimes in $\mathrm{sec}$}

% fitted(n) = c1*n^2 + c2*n^2*2^{N(n)}: blue dashed line, no markers
% (values precomputed using the given N(n) list and fitted constants c1,c2)
\addplot+[
  blue,
  dashed,
  smooth,
  thick,
  mark=none,
] coordinates {
  (200,0.1952849967)
  (400,0.7811399867)
  (600,1.7575649700)
  (800,3.1245599467)
  (1000,4.8821249166)
  (1500,10.9847810624)
  (2000,19.5284996666)
  (2500,30.5132807290)
  (3000,43.9391242498)
  (3500,59.8060302289)
  (4000,78.1139989661)
  (4500,98.8680440881)
  (5000,139.0163089650)
};
\addlegendentry{fitted curve $f(n)=c_1n^2+c_2\,n^2\,2^{N(n)}$}

\addlegendimage{red,thick,mark=x,only marks,mark size=2.5pt}
\addlegendentry{measured mean number of MSPs $N(n)$}

\end{axis}

% --- Right axis: N(n) (red), NO legend here ---
\begin{axis}[
  at={(leftaxis.south west)},
  anchor=south west,
  width=14cm,height=8cm,
  xmin=0,xmax=5200,
  xtick={0,500,...,5000},
  ymin=0,ymax=220,
  ytick={0,20,...,220},
  ylabel={\color{red} $N$},
  axis y line*=right,
  axis x line=none,
]

\addplot+[
  red,
  mark=x,
  mark size=2.5pt,
  smooth,
  thick,
] coordinates {
  (200,28.59) (400,48.37) (600,60.94) (800,73.83) (1000,80.18)
  (1500,106.78) (2000,125.05) (2500,136.87) (3000,153.35)
  (3500,170.37) (4000,185.96) (4500,199.65) (5000,211.07)
};

\end{axis}

\end{tikzpicture}
\caption{The measured mean runtimes of Algorithm~\ref{alg:findallminclust} (blue solid) and the measured mean number of MSPs (red). Notice the two different vertical axes for the blue and red curves.
Also shown is a closed-form fitted function $f(n)=c_1n^2+c_2\,n^2\,2^{N(n)}$ for the mean runtimes (blue dashed).
Thus, the fitting introduces an alternative quadratic term to the theoretical result from Theorem~\ref{thm:complexity}.
Since
$c1 = 4.882E{-}6$
and  $c2 = 1.964E{-}70$,
the second term is negligible for the plotted range of $n$.
}
\label{fig:plot}
\end{figure}

\section{Conclusion}
\label{sec:conclusion}

This paper presented and evaluated a new algorithm for finding minimal clusters in st-DAGs, where a cluster is defined as a special type of subgraph whose entry and exit borders constitute a type of synchronisation point or barrier.
This is a practically relevant problem, since DAGs are used in many areas of application, and the structural analysis of large DAGs is difficult.
Identifying this type of subgraphs enables a divide-and-conquer approach and thereby makes the analysis of large DAGs much more efficient.

Technically, the key point of our approach is the concept of syncpoints, in particular maximum syncpoints (MSPs), where we also described an efficient algorithm of how to find them.
We have introduced a minimal cluster finding algorithm, providing detailed pseudo code and some hints for efficient implementation.
The algorithm doesn't work directly on the given DAG, but on the so-called maximum syncpoint DAG (MSP-DAG), which in general is much smaller.
This is the main source of efficiency of the algorithm.
The paper also contains a
%comprehensive/
rigorous complexity analysis, whose result is supported by a set of non-trivial experiments.

As future work, we are planning to exploit the capabilities of our minimal cluster finding algorithm in the context of a project planning toolset which allows for the prediction of project timelines in a realistic environment with uncertainties and resource limitations.

\phantom{blabla}

\noindent{\bf Acknowledgements:}
The authors would like to cordially thank their colleague Fabian Michel for valuable comments on earlier versions of the manuscript, especially on the complexity analysis.

\appendix
\section{Computing all relevant paths}
\label{sec:computingpaths}

\begin{algorithm}[ht]
    \caption{$ComputeAllVertexPaths$ computes all vertex-paths of any length in a DAG $G=(V,E)$}\label{alg:computeallvpaths}
    \begin{algorithmic}[1]
        \State {\bf function} $ComputeAllVertexPaths$ ($G$: st-DAG) {\bf returns} void
    \For{all $v \in V$}
    \State initialize $PathsTo[v] := [[v]]$
    \State /// initialize path list ending in $v$ with a single path of length zero
    \EndFor
    \For{all $u\in V$ in topological order}
    \label{line:cavp7}
    \For{each successor $v$ of $u$}
    \label{line:cavp8}
    \For{each path $p$ in $PathsTo[u]$}
    \label{line:cavp9}
    \State create a new path $p_{new}$ by appending $v$ to $p$
    \label{line:cavp10}
    \State add $p_{new}$ to the list $PathsTo[v]$
    \label{line:cavp11}
    \EndFor
    \EndFor
    \EndFor
    \State /// now each list $PathsTo[v]$ contains all paths of any length ending in $v$
    \end{algorithmic}
\end{algorithm}

For a given DAG, the set of all vertex-paths (of any length) can be computed by the algorithm $ComputeAllVertexPaths$ shown in Algorithm~\ref{alg:computeallvpaths}. In our work, the algorithm is executed not on the st-DAG $G=(V,E)$, but on the associated MSP-DAG $G^M = (V^M,E^M)$, which usually is much smaller. However, for simplicity we still use the notation $G=(V,E)$ in this appendix. After running the algorithm on $G^M$, only those paths that start in a FSP or FHSP and that end in a FSP or a BHSP are relevant for our purposes, all other paths are not relevant. Therefore, after running the algorithm $ComputeAllVertexPaths$ on $G^M$, we discard all lists $PathsTo[\mathcal{Y}]$ where $\mathcal{Y}$ is a FHSP. Furthermore, in the remaining lists, we delete all paths that start with a BHSP. The resulting lists then contain all relevant vertex paths in $G^M$ whose start and end vertices may constitute the brackets of a cluster.
By construction, each of these lists $PathsTo[\mathcal{Y}]$ is already sorted (in sections) by increasing path length, which is important later for deleting superpaths efficiently.

We now analyze the runtime complexity of algorithm $ComputeAllVertexPaths$ (shown in Algorithm~\ref{alg:computeallvpaths}). The worst-case scenario, i.e.\ the DAG with the maximum number of paths, is a DAG with vertex set $V=\{ 1,\dots, n \}$ and the maximum possible edge set $E=\{(k,l) \; | \; k,l \in V \; \wedge \; l \geq k \}$.
For this graph, the topological sorting  yields the (unique) order $1 \prec 2 \prec \dots \prec n$.
Going through the for-loops of line~\ref{line:cavp7} and line~\ref{line:cavp8}, vertex $u=k$ has $(n-k)$ successors $v$. Vertex $v$ is appended to each path found in $PathsTo[u]$ in line~\ref{line:cavp10}. Since $PathsTo[u]$ contains $\binom{k-1}{j-1}$ paths of length $j-1$, where $j$ is in the range $\{ 1, \dots, k \}$,
and since the new paths $p_{new}$ have length $(j-1)+1= j$,
this requires $\sum_{j=1}^{k} \binom{k-1}{j-1} \cdot j$ elementary operations per successor.
So altogether, the number of operations for executing $ComputeAllVertexPaths$ on this worst-case DAG with $n$ vertices is
\begin{equation}
\sum_{k=1}^{n-1} (n-k)\sum_{j=1}^{k} \binom{k-1}{j-1} \cdot j
 \label{eq:bigcomplex}
\end{equation}
The expression in Equation~\ref{eq:bigcomplex} can be simplified to $(2n-4)2^{n-2}+1$ (by first replacing the inner sum by the expression $(k+1)2^{k-2}$ and then multiplying out and applying geometric series arguments). Thus, the time complexity of Algorithm~\ref{alg:computeallvpaths} is $O(n \cdot 2^n)$. Note that this is the complexity of {\em computing} all paths of any length in a DAG with $n$ vertices (which also involves outputting each of these paths), whereas the {\em number of paths} of any length is only at most $2^n-(n+1)=O(2^n)$, as can be shown by a simple combinatorial argument looking at the worst-case DAG mentioned above.
We have hereby confirmed the well-known result
that the number of paths in a DAG can be exponential in the number of vertices. But it is important to again emphasize that we are applying the algorithm $ComputeAllVertexPaths$ not to the st-DAG but to the usually much smaller MSP-DAG!

%%\end{thebibliography}
\bibliographystyle{elsarticle-num-names}
\bibliography{references.bib}

\end{document}